\documentclass[12pt,a4paper]{article}
\usepackage{amssymb,amsmath}
\usepackage{bbm}
\usepackage{stmaryrd}
\usepackage[mathcal]{euscript}
\usepackage{hyperref}

\numberwithin{equation}{section}

\newcounter{thMM}
\newcounter{leMM}
\newcounter{deFF}
\newcounter{exMP}
\newcounter{prOP}
\newcounter{coRR}
\newenvironment{theorem}[1][Theorem]{\refstepcounter{thMM}\trivlist
   \item[\hskip19pt{\bf #1~\arabic{thMM}.}]\it\hskip3pt}{\endtrivlist}
\newenvironment{lemma}[1][Lemma]{\refstepcounter{leMM}\trivlist
   \item[\hskip19pt{\bf #1~\arabic{leMM}.}]\it\hskip3pt}{\endtrivlist}
\newenvironment{definition}[1][Definition]{\refstepcounter{deFF}\trivlist
   \item[\hskip19pt{\bf #1~\arabic{deFF}.}]\rm\hskip3pt}{\endtrivlist}

\newenvironment{proposition}[1][Proposition]{\refstepcounter{prOP}\trivlist
   \item[\hskip19pt{\bf #1~\arabic{prOP}.}]\it\hskip3pt}{\endtrivlist}
\newenvironment{corollary}[1][Corollary]{\refstepcounter{coRR}\trivlist
   \item[\hskip19pt{\bf #1~\arabic{coRR}.}]\it\hskip3pt}{\endtrivlist}

\newenvironment{proof}[1][Proof]{\begin{trivlist}
\item[\hskip \labelsep {\bfseries #1}] }{ \begin{flushright}$\square$\end{flushright}\end{trivlist}}
\newenvironment{remark}[1][Remark]{\begin{trivlist}
\item[\hskip \labelsep {\bfseries #1}]\it\hskip3pt}{\end{trivlist}}

\newcommand{\Dleft}{[\hspace{-1.5pt}[}
\newcommand{\Dright}{]\hspace{-1.5pt}]}
\newcommand{\SN}[1]{\Dleft #1 \Dright}

\newcommand{\Q}{\mathcal{Q}}

\newcommand{\Id}{\mathbbmss{1}}

\DeclareMathOperator{\Vect}{Vect}

\DeclareMathOperator{\w}{w}

\DeclareMathOperator{\ad}{ad}

\begin{document}
\bibliographystyle{plain}

\author{Andrew James Bruce\\ \small \emph{Pembrokeshire College},\\
\small \emph{Haverfordwest, Pembrokeshire},\\\small  \emph{SA61 1SZ, UK}\\\small email:\texttt{andrewjamesbruce@googlemail.com}}
\date{\today}
\title{Odd Jacobi manifolds and Loday--Poisson brackets}
\maketitle

\begin{abstract}
In this paper we construct a non-skewsymmetric version of a Poisson bracket on the algebra of smooth functions on an odd Jacobi supermanifold. We refer to such Poisson-like brackets as Loday--Poisson brackets. We examine the relations between the Hamiltonian vector fields with respect to both the odd Jacobi structure and the Loday--Poisson structure. Furthermore, we show that the Loday--Poisson bracket satisfies the Leibniz rule over the noncommutative product derived from the homological vector field.
\end{abstract}

\begin{small}
\noindent \textbf{MSC (2010)}: 17B70, 53D10, 53D17, 58A50, 83C47.\\
\noindent \textbf{Keywords}: supermanifolds, Jacobi structures, Q-manifolds, derived brackets, Loday brackets.
\end{small}\\

\noindent\textbf{In memory of Jean-Louis Loday (1946--2012).}

\section{Introduction}\label{sec:Introduction}

In this  paper we construct a Loday--Poisson bracket on the algebra of smooth functions on an odd Jacobi supermanifold $(M, S, Q)$, see \cite{Bruce:2012}. The key observation here is that the homological vector field $Q \in \Vect(M)$ is a Jacobi vector field with respect to the odd Jacobi structure.  Due to this one can directly employ the notion of derived brackets as defined by Kosmann--Schwarzbach \cite{Kosmann--Schwarzbach:1996,Kosmann--Schwarzbach:2004}. We will draw heavily from these works\footnote{The notion of a derived bracket can be traced back to J. L. Koszul (1990) in unpublished notes. The notion was also put forward by Th. Th. Voronov (1993), but remains unpublished.}. \\

On any differential super Lie algebra, $(A,[ , ],D)$, with bracket of even or odd parity, one can construct the ``derived bracket"

\begin{equation}
(a,b)_{D} := (-1)^{\widetilde{a} +1} [D(a),b].
\end{equation}

This bracket is not in general a Lie bracket (or a Lie antibracket) due to the lack of skewsymmetry. However the bracket does satisfy a form of the Jacobi identity. Such Lie algebras mod the skewsymmetry were first introduced by Loday under the name ``Leibniz algebras" \cite{Loday:1992,Loday:1993}, we will shortly be a little more precise here.\\

On an odd Jacobi manifold we have a differential Lie antialgebra, that is an odd Lie bracket provided by the odd Jacobi bracket, and a differential provided by the homological vector field. Then up to some conventions we arrive a the main theorem of this paper:\\

\noindent \textbf{Theorem 1.} \emph{Let $(M, S, Q)$ be an odd Jacobi manifold, then $C^{\infty}(M)$ comes equipped with a canonical Loday--Poisson bracket.}\\

In essence on an odd Jacobi manifold one has ``a Poisson bracket mod the skewsymmetry". We make the above theorem more precise and  then investigate the relations between the Hamiltonian vector fields given by the odd Jacobi and the Loday--Poisson structure. Interestingly, these relations are identical to the well known Cartan identities for the exterior derivative, interior product and Lie derivative acting on differential forms on a manifold.\\

\begin{remark}
Grabowski \& Marmo in \cite{Grabowski:2001B,Grabowski:2003B} prove that on a classical (purely even) manifold $N$ the Jacobi--Loday identity together with the (generalised) Leibniz rule for a binary  bracket on $C^{\infty}(N)$ imply the skewsymmetry: there can only be  Poisson or even Jacobi brackets on classical manifolds. In fact, they show that similar statements also hold for Nambu--Poisson brackets and  for generalisations of Lie algebroid brackets. However, here we consider supermanifolds, meaning that we have a $\mathbbmss{Z}_{2}$-grading and more importantly, non-trivial nilpotent functions which invalidate the conditions of their theorem.
\end{remark}

\noindent \textbf{Preliminaries} \\
All vector spaces and algebras will be $\mathbbmss{Z}_{2}$-graded.   We will generally  omit the prefix \emph{super}. By \emph{manifold} we will mean a \emph{smooth supermanifold}. We denote the Grassmann parity of an object by \emph{tilde}: $\widetilde{A} \in \mathbbmss{Z}_{2}$. By \emph{even} or \emph{odd} we will be referring explicitly to the Grassmann parity.\\

 A \emph{Poisson} $(\varepsilon = 0)$  or \emph{Schouten} $(\varepsilon = 1)$ \emph{algebra} is understood as a vector space $A$ with a bilinear associative multiplication and a bilinear operation $\{\bullet , \bullet\}_{\varepsilon}: A \otimes A \rightarrow A$ such that:
\begin{list}{}
\item \textbf{Grading} $\widetilde{\{a,b \}_{\varepsilon}} = \widetilde{a} + \widetilde{b} + \varepsilon$
\item \textbf{Skewsymmetry} $\{a,b\}_{\varepsilon} = -(-1)^{(\tilde{a}+ \varepsilon)(\tilde{b}+ \varepsilon)} \{b,a \}_{\varepsilon}$
\item \textbf{Jacobi Identity} $\displaystyle\sum\limits_{\textnormal{cyclic}(a,b,c)} (-1)^{(\tilde{a}+ \varepsilon)(\tilde{c}+ \varepsilon)}\{a,\{b,c\}_{\varepsilon}  \}_{\varepsilon}= 0$
\item \textbf{Leibniz Rule} $\{a,bc \}_{\varepsilon} = \{a,b \}_{\varepsilon}c + (-1)^{(\tilde{a} + \varepsilon)\tilde{b}} b \{a,c \}_{\varepsilon}$
\end{list} \vspace{10pt}
for all homogenous elements $a,b,c \in A$. Extension to inhomogeneous elements is via linearity. If the Leibniz rule does not hold identically, but is modified as
\begin{equation}
\{a,bc \}_{\varepsilon} = \{a,b \}_{\varepsilon}c + (-1)^{(\tilde{a} + \varepsilon)\tilde{b}} b \{a,c \}_{\varepsilon} - \{a ,\mathbbmss{1}  \}_{\varepsilon} bc,
\end{equation}

then we have even ($\epsilon = 0)$ or odd ($\epsilon = 1)$ \emph{Jacobi algebras}.\\

A manifold $M$ such that $C^{\infty}(M)$ is a Poisson/Schouten algebra is known as a \emph{Poisson/Schouten manifold}. In particular the cotangent of a manifold comes equipped with a canonical Poisson structure.\\

Let us employ   natural local coordinates $(x^{A}, p_{A})$ on $T^{*}M$, with $\widetilde{x}^{A} = \widetilde{A}$ and $\widetilde{p}_{A} = \widetilde{A}$. Local diffeomorphisms on $M$ induce vector  bundle automorphism on $T^{*}M$ of the form
\begin{equation}
\nonumber \overline{x}^{A} = \overline{x}^{A}(x), \hspace{30pt} \overline{p}_{A}  = \left(\frac{\partial x^{B}}{\partial \overline{x}^{A}}\right)p_{B}.
\end{equation}

We will in effect use the local description as a \emph{natural vector bundle} to define the cotangent bundle of a supermanifold.  The canonical Poisson bracket on the cotangent is given by

\begin{equation}
\{ F,G \} = (-1)^{\widetilde{A} \widetilde{F} + \widetilde{A}} \frac{\partial F}{\partial p_{A}}\frac{\partial G}{\partial x^{A}} - (-1)^{\widetilde{A}\widetilde{F}}\frac{\partial  F}{\partial x^{A}} \frac{\partial G}{\partial p_{A}}.
\end{equation}\\

A manifold equipped with an odd vector field $Q$, such that the non-trivial condition $Q^{2}= \frac{1}{2}[Q,Q]=0$ holds, is known as a \emph{Q-manifold}. The vector field $Q$ is known as a \emph{homological vector field} for obvious reasons. \\

Important in this paper is the notion of a Loday--Poisson bracket. We define such a bracket being ``a Poisson bracket but without the skewsymmetry". Note that the Jacobi identity, as presented above, for an even bracket (with or without the Leibniz rule) implies that the  adjoint endomorphism $f \rightsquigarrow \ad_{f} := \{f, \bullet \}$ is a derivation over the bracket itself.  That is
\begin{equation*}
\ad_{f}\{g,h \} =  \{\ad_{f}(g) ,h\} + (-1)^{\widetilde{f} \widetilde{g}} \{g, \ad_{f}(h)  \}.
\end{equation*}

Note that the above has direct meaning even if the skewsymmetry of the bracket is weakened. Thus, the most useful definition of the Jacobi identity  for an even bracket when the symmetry is lost is
\begin{equation}
\{f, \{g,h \}  \} = \{ \{f,g \},h  \} + (-1)^{\widetilde{f} \widetilde{g}}\{ g, \{f,h  \} \},
\end{equation}

We refer to this form of the Jacobi identity as the \emph{Jacobi--Loday identity} \cite{Loday:1992}. If we have an even bracket that satisfies the Jacobi--Loday identity and the (left) Leinbiz rule, but not necessarily the skewsymmetry then we call such a bracket a \emph{Loday--Poisson bracket}. For example, classical Poisson brackets on a manifold are specific examples of Loday--Poisson bracket, albeit skewsymmetric. Furthermore, we will not insist that a Loday--Poisson bracket is a bi-derivation, only that it satisfies the left Leibniz rule.

\begin{remark}
Much of this paper will generalise directly to QS-manifolds in the sense of Voronov \cite{Voronov:2001qf}. A primary example of  a QS-manifold is the antitangent bundle $\Pi TM$ of a Poisson manifold $M$. The Schouten structure is supplied by the Koszul--Schouten bracket and the homological vector field by the de Rham differential \cite{Koszul:1985}. The associated Loday--Poisson bracket is a natural extension of the Poisson bracket on $C^{\infty}(M)$ to differential forms over $M$, but note this extension is as a Loday bracket only \cite{Kosmann--Schwarzbach:1996}.  The constructions here will not directly generalise to even Jacobi supermanifolds due to incompatibility of the Grassmann parities. There is no canonical choice of a  homological Jacobi vector field on a even Jacobi supermanifold. On a classical Jacobi manifold there are no (non zero) odd vector fields at all.
\end{remark}

\newpage
\section{Odd Jacobi manifolds}

Lichnerowicz \cite{Lichnerowicz1977} introduced the notion of a Poisson manifold as well as a Jacobi manifold. Such  manifolds have found applications in classical mechanics and play an important role in quantisation. In this section we recall some of the basic notions as pertaining to odd Jacobi (super)manifolds. No proofs are given here and can be found in \cite{Bruce:2012}, or follow directly from the definitions.  For a modern discussion of Jacobi structures, including $\mathbbmss{Z}$-graded versions see \cite{Grabowski:2001,Grabowski:2003,Grabowski:2011}.

\begin{definition}
An \textbf{odd Jacobi structure} $(S,Q)$ on a manifold $M$  consists of
\begin{itemize}
\item an odd function $S \in C^{\infty}(T^{*}M)$, of degree two in fibre coordinates,
\item an odd vector field $Q \in \Vect(M)$,
\end{itemize}
such that the following conditions hold:
\begin{enumerate}
\item the homological condition $Q^{2} = \frac{1}{2} [Q,Q]=0$,
\item the invariance condition  $L_{Q}S = 0$,
\item the compatibility condition $\{S,S \}= - 2 \Q S $.
\end{enumerate}
 Here $\Q \in C^{\infty}(T^{*}M)$ is the principle symbol or ``Hamiltonian"  of the vector field $Q$. The brackets $\{ \bullet, \bullet \}$ are the canonical Poisson brackets on the cotangent bundle of the manifold.\\
A manifold equipped with an odd Jacobi structure $(S,Q)$ is known as an \textbf{odd Jacobi manifold}.
\end{definition}

\begin{definition}
The \textbf{odd Jacobi bracket} on $C^{\infty}(M)$ is defined as
\begin{eqnarray}
\SN{f,g}_{J} &=&  (-1)^{\widetilde{f}+1} \{ \{ S,f \},g    \} - (-1)^{\widetilde{f}+1} \{ \Q, fg \}\\
\nonumber &=&(-1)^{(\widetilde{B}+1)\widetilde{f}  +1} S^{BA} \frac{\partial f}{\partial x^{A}} \frac{\partial g}{\partial x^{B}} + (-1)^{\widetilde{f}} \left(Q^{A} \frac{\partial f}{\partial x^{A}}   \right)g  + f \left( Q^{A}\frac{\partial g}{\partial x^{A}}  \right),
\end{eqnarray}
with $f,g \in C^{\infty}(M)$.
\end{definition}

The odd Jacobi bracket makes the algebra of smooth functions on $M$ into an odd Jacobi algebra.

\begin{remark}
The definition of an odd Jacobi manifold given here is not quite the most general, one can include an odd function in the construction of an odd Jacobi structure, see Grabowski and Marmo \cite{Grabowski:2003} for details. Note that the definition of given by Grabowski and Marmo  coincide with that given here (up to conventions) up on setting the odd function to zero. Thus, there are examples of odd Jacobi brackets on supermanifolds not covered by the constructions here.
\end{remark}

\begin{definition}
Given a function $f \in C^{\infty}(M)$ the associated \textbf{Hamiltonian vector field} is given by
\begin{eqnarray}
f &\rightsquigarrow & X_{f} \in \Vect(M)\\
\nonumber X_{f}(g) &=& (-1)^{\widetilde{f}} \SN{f,g}_{J} - Q(f) g.
\end{eqnarray}
\end{definition}

Note that the homological vector field $Q$ is itself Hamiltonian with respect to the  unit constant function: $Q = X_{\mathbbmss{1}} = \SN{\mathbbmss{1}, \bullet}_{J}$.

\begin{definition}
A vector field $X \in \Vect(M)$ is said to be a \textbf{Jacobi vector field} if and only if
\begin{equation}
L_{X}S =  \{ \chi, S \} = 0 \hspace{35pt} \textnormal{and}  \hspace{15pt} L_{X}Q = \{\chi, \Q  \} =  0,
\end{equation}
where $\chi \in C^{\infty}(T^{*}M)$ is the symbol or ``Hamiltonian" of the vector field $X$. Note that the homological vector field $Q$ is a Jacobi vector field.
\end{definition}

\newpage

\begin{proposition}\label{prop:jacobivectorfields}
Let $X \in \Vect(M)$ be a vector field on an odd Jacobi manifold. Then the following are equivalent:
\begin{enumerate}
\item $X$ is a Jacobi vector field.
\item $X$ is a derivation over the odd Jacobi bracket;
\begin{equation}\nonumber
X(\SN{f,g}_{J}) = \SN{X(f), g}_{J} + (-1)^{\widetilde{X}(\widetilde{f}+1)} \SN{f,X(g)}_{J}.
\end{equation}
\item $[X, Y_{f}] = (-1)^{\widetilde{X}} Y_{X(f)}$, for all Hamiltonian vector fields $Y_{f}$.
\end{enumerate}
\end{proposition}

\begin{proposition}\label{prop:Hamiltonianvectorfields}
A Hamiltonian vector field $X_{f} \in \Vect(M)$ is a Jacobi vector field if and only if $f \in C^{\infty}(M)$ is Q-closed.
\end{proposition}

\begin{proposition}\label{prop:morphism}
The assignment $f \rightsquigarrow X_{f}$ is a morphism  between the odd Lie algebra on $C^{\infty}(M)$ provided by the odd Jacobi brackets and the Lie algebra of vector fields. Specifically, the following holds:
\begin{equation*}
[X_{f}, X_{g}] = - X_{\SN{f,g}_{J}}
\end{equation*}
for all $f,g \in C^{\infty}(M)$.
\end{proposition}

\begin{proposition}\label{prop:product}
On an odd Jacobi manifold the following identity holds:
\begin{equation*}
X_{fg} = (-1)^{\widetilde{f}} f \: X_{g} + (-1)^{\widetilde{g}(\widetilde{f}+1)}g \: X_{f} + (-1)^{\widetilde{f} + \widetilde{g}+1} f g \: Q.
\end{equation*}
\end{proposition}

\section{Loday--Poisson brackets}\label{sect:mainresult}
We are now in a position to state the main theorem of this paper. We draw heavily on the notion of derived brackets in the sense of Kosmann--Schwarzbach,  see \cite{Kosmann--Schwarzbach:1996,Kosmann--Schwarzbach:2004} for details.

\begin{theorem}\label{theo:main}
Let $(M, S, Q)$ be an odd Jacobi manifold, then $C^{\infty}(M)$ comes equipped with a canonical Loday--Poisson bracket.
\end{theorem}
\begin{proof}
We prove the theorem by direct construction of the Loday--Poisson bracket. The bracket is canonical up to minor issues of conventions. Following Kosmann--Schwarzbach \cite{Kosmann--Schwarzbach:1996} we define the Loday--Poisson bracket as
\begin{equation*}
\{f,g \}_{J} :=  (-1)^{\widetilde{f}+1}  \SN{Q(f),g}_{J}.
\end{equation*}
As the vector field $Q$ is homological and Jacobi, together with the Jacobi identify for the odd Jacobi bracket the above bracket is \emph{even} and satisfies the  \emph{Jacobi--Loday identity}. For completeness we outline the steps here and urge the reader to consult  \cite{Kosmann--Schwarzbach:1996}.\\

From the definitions and the Jacobi identity for the odd Jacobi bracket we have
\begin{eqnarray}
\nonumber \{f, \{g,h  \}_{J} \}_{J} &=& (-1)^{\widetilde{f} + \widetilde{g}} \SN{Q(f) , \SN{Q(g),h}_{J}}_{J}\\
\nonumber &=& (-1)^{\widetilde{f} +\widetilde{g}} \left( \SN{\SN{Q(f), Q(g)}_{J},h}_{J}+(-1)^{\widetilde{f} \widetilde{g}} \SN{Q(g), \SN{Q(f), h}_{J}}_{J} \right).
\end{eqnarray}

Using the fact that the homological vector field is a Jacobi vector field we have

\begin{eqnarray}
\nonumber \{f, \{g,h  \}_{J} \}_{J} &=& (-1)^{\widetilde{g}+ \widetilde{f} + 1} \SN{(-1)^{\widetilde{f}+1}Q\left( \SN{Q(f), g}_{J} \right), h}_{J} + (-1)^{\widetilde{f}\widetilde{g} +\widetilde{f}+ \widetilde{g}} \SN{Q(g), \SN{Q(f), h}_{J}}_{J} \\
\nonumber &-& (-1)^{\widetilde{g}} \SN{\SN{Q^{2}(f), g}_{J}, h}_{J}.
\end{eqnarray}
As $Q^{2}=0$ the above  gives

\begin{equation*}
\{f, \{g,h \}_{J}  \}_{J} = \{ \{f,g \}_{J},h  \}_{J} + (-1)^{\widetilde{f} \widetilde{g}}\{ g, \{f,h  \}_{J} \}_{J}.
\end{equation*}

Note however, the bracket is not automatically skewsymmetric.  The Leibniz rule follows from the definitions directly:
\begin{eqnarray}
\nonumber \{f, gh \}_{J} &=& (-1)^{\widetilde{f}+1}\SN{Q(f), gh}_{J},\\
\nonumber &=& (-1)^{\widetilde{f}+1} \SN{Q(f), g}_{J}h + (-1)^{\widetilde{f}+1 + \widetilde{f} \widetilde{g}}g \SN{Q(f),h}_{J},\\
\nonumber &=& \{ f,g \}_{J} h + (-1)^{\widetilde{f} \widetilde{g}}g \{ f,h \}_{J} - \SN{Q(f), \Id}_{J}gh,
\end{eqnarray}
which follows from the modified Leibniz rule for the odd Jacobi bracket and the definition of the Loday--Poisson bracket. Then as $\SN{f, \Id}_{J} = (-1)^{\widetilde{f}}Q(f)$ we have $\SN{Q(f), \Id}_{J} = \pm Q(Q(f)) =0$ as $Q$ is homological. This could also be shown via direct application of the Jacobi identity for the odd Jacobi bracket. Then
\begin{equation*}
 \{f, gh \}_{J} = \{ f,g \}_{J} h + (-1)^{\widetilde{f} \widetilde{g}}g \{ f,h \}_{J}.
\end{equation*}
\end{proof}

\begin{remark}
Note that the centre of the odd Jacobi algebra consists entirely of $Q$-closed functions. This is evident as we have $\SN{f, \Id}_{J}=(-1)^{\widetilde{f}}Q(f)$. Thus, the restriction of the  Loday--Poisson bracket to the centre of $\left(C^{\infty}(M), \SN{\bullet,\bullet}_{J}\right)$ is identically zero, or in other words the trivial Poisson bracket.
\end{remark}

In local coordinates the Loday--Poisson bracket is given by

\begin{eqnarray}
\nonumber \{f,g  \}_{J} &=& (-1)^{\widetilde{B}(\widetilde{f}+1) +1} \left((-1)^{\widetilde{A}} S^{BA}Q^{C} \frac{\partial^{2} f}{\partial x^{C} \partial x^{A}} + S^{BA}\frac{\partial Q^{C}}{\partial x^{A}}\frac{\partial f}{\partial x^{C}} \right. \\
  &-&\left.  Q^{B} Q^{A} \frac{\partial f}{\partial x^{A}}\right)\frac{\partial g}{\partial x^{B}}.
\end{eqnarray}

Directly from this local expression  we see  that the bracket satisfies the left Leibniz rule  and is not skewsymmetric, in general.  From the definitions it is clear that

\begin{equation}
\{f,g \}_{J} - (-1)^{\widetilde{f} \widetilde{g}} \{g, f  \}_{J} = (-1)^{\widetilde{f}+1}Q \left( \SN{f,g}_{J} \right).
\end{equation}

\newpage
\section{Some simple examples}

In this section we briefly present four examples of odd Jacobi manifolds and the Loday--Poisson brackets associated with them. These examples are taken straight from \cite{Bruce:2012}.

\subsection*{Schouten manifolds}
Schouten manifolds can be considered as odd Jacobi manifolds with the homological vector field being the zero vector. In this case the corresponding Loday--Poisson bracket is also trivial,
\begin{equation*}
\{f,g \}_{J} = 0,
\end{equation*}
for all $f,g \in C^{\infty}(M)$. Examples of Schouten manifolds include odd symplectic manifolds, which have found applications in physics via the Batalin--Vilkovisky formalism.

\subsection*{Q-manifolds}
Q-manifolds are understood as odd Jacobi manifolds with the almost Schouten structure being trivial, $S=0$. Q-manifolds have found important applications in the Batalin--Vilkovisky formalism along side Schouten structures, \cite{Alexandrov:1997,Schwartz:1993}. The  odd Jacobi bracket  on a Q-manifold is then given by

\begin{eqnarray}
\nonumber \SN{f,g}_{Q} &=& (-1)^{\widetilde{f}} Q(fg)\\
\nonumber &=& (-1)^{\widetilde{f}}Q^{A} \frac{\partial f}{\partial x^{A}}g  + f Q^{A} \frac{\partial g}{\partial x^{A}}.
\end{eqnarray}

The associated Loday--Poisson bracket is thus

\begin{equation}
\{f,g \}_{Q} = (-1)^{\widetilde{f}+1}Q(f)Q(g),
\end{equation}

which is in fact skewsymmetric and thus a genuine Poisson bracket. One can see this directly or from the fact that $Q(\SN{f,g}_{Q})=0$.  In local coordinates we have

\begin{equation}
\{f,g \}_{Q} = (-1)^{\widetilde{B}(\widetilde{f}+1)}Q^{B}Q^{A}\frac{\partial f}{\partial x^{A}}\frac{\partial g}{\partial x^{B}}.
\end{equation}

\begin{remark}
There is nothing really new here. In essence all we have is a Poisson structure (bi-vector) given by $P =  \pm \frac{1}{2}(\varsigma Q)^{2} $ where $\varsigma : \Vect(M) \rightarrow C^{\infty}(\Pi T^{*}M)$ is the odd isomorphism between vector fields and ``one-vectors".  Note that $\varsigma Q$ is now even and the homological property becomes $\SN{\varsigma Q, \varsigma Q}=0$, where the bracket here is the canonical Schouten--Nijenhuis bracket. Thus due to the Leibniz rule is is clear that $\SN{P,P} =0$ and we have a Poisson structure.  One could also build higher Poisson structures from Q-manifolds in this way.
\end{remark}

In the  next example we consider Lie algebroids as Q-manifolds with an extra grading, known as weight, on the local coordinates.

\section*{Lie algebroids}

Recall that a Lie algebroid $E \rightarrow M$ can be described as a weight one homological vector field on the total space of $\Pi E$. This understanding in terms of a homological vector field is attributed to Va$\breve{\textrm{{\i}}}$ntrob \cite{Vaintrob:1997}. The weight is assigned as zero to the base coordinates and one to the fibre coordinates. In natural local coordinates $(x^{A}, \xi^{\alpha})$ the homological vector field is of the form

\begin{equation*}
Q = \xi^{\alpha}Q_{\alpha}^{A}(x) \frac{\partial}{\partial x^{A}} + \frac{1}{2} \xi^{\alpha} \xi^{\beta} Q_{\beta \alpha}^{\gamma}(x)\frac{\partial}{\partial \xi^{\gamma}} \in \Vect(\Pi E).
\end{equation*}

The associated weight one odd Jacobi bracket is given by

\begin{eqnarray}
\nonumber \SN{\phi, \psi}_{E} &=& (-1)^{\widetilde{\phi}} \left( \xi^{\alpha}Q_{\alpha}^{A}(x) \frac{\partial \phi}{\partial x^{A}} + \frac{1}{2} \xi^{\alpha} \xi^{\beta} Q_{\beta \alpha}^{\gamma}(x)\frac{\partial \phi}{\partial \xi^{\gamma}} \right)\psi\\
  &+& \phi \left( \xi^{\alpha}Q_{\alpha}^{A}(x) \frac{\partial \psi}{\partial x^{A}} + \frac{1}{2} \xi^{\alpha} \xi^{\beta} Q_{\beta \alpha}^{\gamma}(x)\frac{\partial \psi}{\partial \xi^{\gamma}}  \right),
\end{eqnarray}

where $\phi, \psi \in C^{\infty}(\Pi E)$ are ``Lie algebroid differential forms". The weight two Poisson bracket is given in local coordinates by

\begin{eqnarray}
\nonumber \{ \phi, \psi \}_{E} &=& (-1)^{\widetilde{B}\widetilde{\phi} + (\widetilde{B}+1)\widetilde{\gamma}} \xi^{\gamma} \xi^{\alpha}Q_{\alpha}^{B}Q_{\gamma}^{A}\frac{\partial \phi}{\partial x^{A}}\frac{\partial \psi}{\partial x^{B}}\\
\nonumber &+&\frac{1}{2}\xi^{\gamma}\xi^{\delta} \xi^{\alpha}\left( (-1)^{\widetilde{B}(\widetilde{\phi}+1) +(\widetilde{\gamma}+ \widetilde{\delta})(\widetilde{B}+1)} Q^{B}_{\alpha}Q_{\delta \gamma}^{\epsilon} \frac{\partial \phi}{\partial \xi^{\epsilon}}\frac{\partial \psi}{\partial x^{B}}\right.\\
\nonumber &+& \left. (-1)^{(\widetilde{\epsilon} +1)(\widetilde{\phi}+1) + \widetilde{\epsilon}(\widetilde{\gamma}+1)}Q_{\alpha \delta}^{\epsilon}Q_{\gamma}^{B}\frac{\partial \phi}{\partial x^{B}}\frac{\partial \psi}{\partial \xi^{\epsilon}}\right)\\
&+& (-1)^{(\widetilde{\beta}+1)(\widetilde{\phi}+1) + \widetilde{\beta}(\widetilde{\epsilon}+ \widetilde{\rho})}\frac{1}{4}\xi^{\epsilon}\xi^{\rho}\xi^{\gamma}\xi^{\delta}Q_{\delta \gamma}^{\beta}Q_{\rho \epsilon}^{\alpha}\frac{\partial \phi}{\partial \xi^{\alpha}} \frac{\partial \psi}{\partial \xi^{\beta}}.
\end{eqnarray}

\begin{remark}
The constructions here directly generalise to $L_{\infty}$-algebroids, understood as the pair $(\Pi E, Q)$, where $Q \in \Vect(\Pi E)$ is a homological vector field now inhomogeneous in weight.
\end{remark}

\subsection*{Odd contact manifolds}
Consider the manifold $M= \Pi T^{*}N \times \mathbbmss{R}^{0|1}$, where $N$ is a pure even classical manifold. Let us equip this manifold with natural local coordinates $(x^{a}, x^{*}_{a}, \tau)$, where $(x^{*}_{a})$ are fibre coordinates on $\Pi T^{*}N$, which are Grassmann odd and $\tau$ is the coordinate on  the factor  $\mathbbmss{R}^{0|1}$.\\

The manifold $M$ is an odd contact manifold. That is it comes with an odd contact form given by
\begin{equation*}
\alpha = d \tau - x^{*}_{a}dx^{a}.
\end{equation*}

Associated with this odd contact form is an odd Jacobi structure given by
\begin{eqnarray}
\nonumber S &=& p_{*}^{a}(p_{a} + x^{*}_{a}\pi ) \in C^{\infty}(T^{*}M),\\
\nonumber Q &=& - \frac{\partial}{\partial \tau} \in \Vect(M),
\end{eqnarray}

where we have employed fibre coordinates $(p_{a}, p_{*}^{a}, \pi)$ on $T^{*}M$. The corresponding odd Jacobi bracket is given by

\begin{eqnarray}
\nonumber \SN{f,g}_{J} &=& (-1)^{\widetilde{f}+ 1} \frac{\partial f}{\partial x^{*}_{a}} \frac{\partial g}{\partial x^{a}} - \frac{\partial f}{\partial x^{a}} \frac{\partial g}{\partial x^{*}_{a}}\\
\nonumber &+& x^{*}_{a} \frac{\partial f}{\partial x^{*}_{a}} \frac{\partial g}{\partial \tau} - (-1)^{\widetilde{f} + 1} \frac{\partial f}{\partial \tau} x^{*}_{a}\frac{\partial g}{\partial x^{*}_{a}}\\
\nonumber &+& f \frac{\partial g}{\partial \tau} - (-1)^{\widetilde{f}+1} \frac{\partial f}{\partial \tau} g.
\end{eqnarray}

Then the Loday--Poisson bracket is given by

\begin{eqnarray}
\{ f,g \}_{J} &=& \frac{\partial^{2}f}{\partial x^{*}_{a}\partial \tau} \frac{\partial g}{\partial x^{a}} - (-1)^{\widetilde{f}} \frac{\partial^{2}f}{\partial x^{a} \partial \tau} \frac{\partial g}{\partial x^{*}_{a}}\\
\nonumber &+& (-1)^{\widetilde{f}} x^{*}_{a} \frac{\partial^{2}f }{\partial x^{*}_{a} \partial \tau}\frac{\partial g}{\partial \tau} + (-1)^{\widetilde{f}} \frac{\partial f}{\partial \tau} \frac{\partial g}{\partial \tau}.
\end{eqnarray}

\section{Hamiltonian vector fields}
We  now continue this paper with a study of the algebraic relations between the Hamiltonian vector fields with respect to the odd Jacobi bracket and the (even) Loday--Poisson bracket.

\begin{definition}\label{def:Ham}
Let $f \in C^{\infty}(M)$ be an arbitrary function. The associated \textbf{Hamiltonian vector field} with respect to the Loday--Poisson bracket $Y_{f} \in \Vect(M)$ is defined viz
\begin{equation*}
Y_{f}(g) = \{f,g  \}_{J}
\end{equation*}
\end{definition}

Throughout this section we will denote Hamiltonian vector fields with respect to the odd Jacobi structure as $X_{f}$ and those with respect to the Loday--Poisson structure as $Y_{f}$ in order to distinguish the two.  Note that $\widetilde{X_{f}} = \widetilde{f} +1$ and that $\widetilde{Y_{f}} = \widetilde{f}$.

\begin{lemma}\label{lem:Bihamiltonian}
Let $Y_{f}$ be the Hamiltonian vector fields associated with the function $f \in C^{\infty}(M)$ with respect to the  Loday--Poisson bracket. Then
\begin{equation*}
Y_{f} = X_{Q(f)} = - \left[  Q, X_{f}\right],
\end{equation*}
where $X_{f}$ is the Hamiltonian vector field associated with $f$ with respect to the odd Jacobi structure.
\end{lemma}

\begin{proof}
From the definitions
\begin{equation*}
 Y_{f}(g) = \{f,g  \}_{J} = (-1)^{\widetilde{f}+1} \SN{Q(f), g}_{J}= X_{Q(f)}(g) + Q\left(  Q(f)\right)g\\
\end{equation*}

then given that $Q$ is homological we get $Y_{f} = X_{Q(f)}$. Then using Proposition \ref{prop:jacobivectorfields}. we get $X_{Q(f)} = - \left[Q, X_{f} \right]$ which established the Lemma.
\end{proof}

The above lemma can be viewed as establishing a mild generalisation of bi-Hamiltonian systems. In particular, any vector field that is Hamiltonian with respect to the Loday--Poisson bracket is also Hamiltonian with respect to the odd Jacobi structure and  the Hamiltonians are related directly via the homological field.

\begin{corollary}
Let $f \in C^{\infty}(M)$ be an even function that satisfies the ``classical master equation" $\SN{f,f}_{J}=0$. Then this implies $\{f,f \}_{J}=0$. Furthermore we have $\SN{Q(f),f}_{J} =0$ and $\{f ,Q(f)\}_{J} =0$.
\end{corollary}

The above corollary naturally generalises the statement that for classical bi-Hamiltonian systems both Hamiltonians are in involution with respect to both the Poisson structures. Also, note that the Hamiltonian vector fields with respect to the Loday--Poisson bracket only depend on the Q-cohomology class of the Hamiltonian function. Specifically, if $f-f' = Q(g)$ for for some $g \in C^{\infty}(M)$ then $Y_{f}= Y_{f'}$.

\begin{proposition}\label{prop:Ham1}
Let $Y_{f}$ be the Hamiltonian vector field associated with the function \newline $f \in C^{\infty}(M)$ with respect to the Loday--Poisson bracket. Then
\begin{equation*}
\left[ Q, Y_{f} \right] =0.
\end{equation*}
\end{proposition}

\begin{proof}
 Consider an arbitrary  function $g \in C^{\infty}(M)$. Directly we have
\begin{eqnarray}
\nonumber Q \left( Y_{f}(g) \right) &=& Q \{ f,g \}_{J},\\
\nonumber &=& (-1)^{\widetilde{f}+1}Q \SN{Q(f), g}_{J}= - \SN{Q(f), Q(g)}_{J},\\
\nonumber &=& (-1)^{\widetilde{f}} \{ f , Q(g)\}_{J} = (-1)^{\widetilde{f}} Y_{f}(Q(g)).
\end{eqnarray}
Thus, $Q \circ  Y_{f} - (-1)^{\widetilde{f}} Y_{f} \circ Q =0$.
\end{proof}

\begin{lemma}\label{lem:1}
Let $Y_{f}$ be the Hamiltonian vector field associated with the function $f \in C^{\infty}(M)$ with respect to the Loday--Poisson bracket. Then $Y_{f}$ is a Jacobi vector field.
\end{lemma}

\begin{proof}
Follows from Lemma \ref{lem:Bihamiltonian}. and that Hamiltonian vector fields with respect to the odd Jacobi bracket are Jacobi if and only if the Hamiltonian function is $Q$-closed. As $Q^{2}=0$ evidently $Y_{f}$ is Jacobi.
\end{proof}

\begin{proposition}\label{prop:Ham2}
Let $Y_{f}$ and $Y_{g}$ be the Hamiltonian vector fields associated with the functions $f,g \in C^{\infty}(M)$ with respect to the Loday--Poisson bracket. Then
\begin{equation*}
\left[Y_{f},  Y_{g} \right] = Y_{\{f,g \}_{J}}
\end{equation*}
\end{proposition}

\begin{proof}
Via direct computation
\begin{equation*}
\left[Y_{f}, Y_{g}  \right] = \left[X_{Q(f)}, X_{Q(g)}  \right]= - X_{\SN{Q(f), Q(g)}_{J} }\\
\end{equation*}
using the properties of  Hamiltonian vector fields associated with the odd Jacobi bracket. Then using
\begin{equation*}
Q\left( \SN{Q(f), g}_{J}\right) = (-1)^{\widetilde{f}}\SN{Q(f), Q(g)}_{J},
\end{equation*}
we arrive at
\begin{equation*}
- X_{\SN{Q(f), Q(g)}_{J} } = (-1)^{\widetilde{f}+1}X_{Q \left(\SN{Q(f),g}_{J}  \right)},
\end{equation*}
which established the proposition.
\end{proof}

The above proposition is rather expected and more interesting are ``mixed" commutators of the Hamiltonian vector fields.  In particular are their nice expression for $X_{\{f,g \}_{J}}$ and $Y_{\SN{f,g}_{J}}$?

\begin{proposition}\label{prop:Ham3}
Let $X_{g}$ and $Y_{f}$ be the Hamiltonian vector fields associated with the functions $g,f \in C^{\infty}(M)$ with respect to the odd Jacobi bracket and the Loday--Poisson bracket respectively. Then
\begin{equation*}
\left[ Y_{f}, X_{g} \right] = (-1)^{\widetilde{f}}X_{\{f,g\}_{J}}.
\end{equation*}
\end{proposition}
\begin{proof}
Follows from the fact that $Y_{f}$ is a Jacobi vector field, see Lemma \ref{lem:1}.
\end{proof}

\begin{proposition}\label{prop:Ham2}
Let $X_{f}$ and $Y_{f}$ be the Hamiltonian vector fields associated with the function $f \in C^{\infty}(M)$ with respect to the odd Jacobi bracket and the Loday--Poisson bracket respectively. Then
\begin{equation*}
(-1)^{\widetilde{f}+1} Y_{\SN{f,g}_{J}} = X_{\{f,g\}_{J}} + (-1)^{\widetilde{f}\widetilde{g}} X_{\{g,f\}_{J}}.
\end{equation*}
\end{proposition}
\begin{proof}
Using the fact that $Y_{f} = X_{Q(f)}$ and that $Q$ is Jacobi we see that
\begin{eqnarray}
\nonumber Y_{\SN{f,g}_{J}} &=& X_{Q \SN{f,g}_{J}}\\
\nonumber &=& X_{\SN{Q(f),g}_{J}} + (-1)^{\widetilde{f}+1} X_{\SN{f,Q(g)}_{J}}\\
\nonumber &=& (-1)^{\widetilde{f}+1}X_{\{f,g \}_{J}} + {-1}^{\widetilde{f} + \widetilde{g}(\widetilde{f}+1)}X_{\SN{Q(g),f}_{J}}\\
\nonumber &=& (-1)^{\widetilde{f}+1}X_{\{f,g \}_{J}} + (-1)^{\widetilde{f}+1 + \widetilde{f}\widetilde{g}}X_{\{g,f \}_{J}}.
\end{eqnarray}
Thus multiplying through by the appropriate sign factor we arrive at
\begin{equation*}
(-1)^{\widetilde{f}+1} Y_{\SN{f,g}_{J}} = X_{\{f,g\}_{J}} + (-1)^{\widetilde{f}\widetilde{g}} X_{\{g,f\}_{J}}.
\end{equation*}
\end{proof}

\begin{corollary}\label{coro:QCommutators}
With the definitions previously given
\begin{enumerate}
\item $\left [ Q, X_{\{f,g \}_{J}} \right] = - \left [ Y_{f}, Y_{g}\right] = - Y_{\{f,g\}_{J}}$.
\item $\left[ Q, Y_{\SN{f,g}_{J}}\right]=0$.
\end{enumerate}
\end{corollary}

Expressions for higher nested commutators of Hamiltonian vector field can be worked out from the relations given here and the Jacobi identity for the commutator.\\

The final thing to consider in this section is  how the Hamiltonian vector field with respect to the Loday--Poisson bracket behaves under the product of two functions.

\begin{proposition}\label{prop:LodayProduct}
On an odd Jacobi manifold the following identity holds:
\begin{eqnarray}
\nonumber Y_{fg} &=& f \: Y_{g} + (-1)^{\widetilde{f}  \widetilde{g}} g \: Y_{f}\\
\nonumber &+& (-1)^{\widetilde{f} +1} Q(f) \left( X_{g} - (-1)^{\widetilde{g} }  g \: Q\right)\\
\nonumber &+& (-1)^{\widetilde{f}  \widetilde{g} + \widetilde{g} +1} Q(g) \left( X_{f} - (-1)^{\widetilde{f} }  f \: Q\right).
\end{eqnarray}
\end{proposition}

\begin{proof}
First note from Lemma \ref{lem:Bihamiltonian}. and the Leibniz rule for $Q$ that
\begin{equation*}
Y_{fg} = X_{Q(fg)} = X_{Q(f)g} + (-1)^{\widetilde{f}} X_{f Q(g)}.
\end{equation*}
Then application of Proposition  \ref{prop:product}. produces
\begin{eqnarray}
\nonumber Y_{fg} &=& f \: X_{Q(g)} + (-1)^{\widetilde{f}  \widetilde{g}} g \: X_{Q(f)}\\
\nonumber &+& (-1)^{\widetilde{f} +1} Q(f) \left( X_{g} - (-1)^{\widetilde{g} }  g \: Q\right)\\
\nonumber &+& (-1)^{\widetilde{f}  \widetilde{g} + \widetilde{g} +1} Q(g) \left( X_{f} - (-1)^{\widetilde{f} }  f \: Q\right),
\end{eqnarray}
which implies the proposition.
\end{proof}

It is easy to verify the ``consistency conditions" $Y_{\Id g} = Y_{g}$ and $Y_{f \Id}= Y_{f}$. Furthermore, Proposition \ref{prop:LodayProduct}. can be interpreted as ``measuring" the violation of the right Leibniz rule  of Loday--Poisson bracket. Specifically we have

\begin{eqnarray}
\{fg,h \}_{J} &=& f \{g,h  \}_{J} + (-1)^{\widetilde{g} \widetilde{h}}\{f,h \}_{J} g\\
\nonumber &+& (-1)^{\widetilde{f}+1}Q(f)\left( (-1)^{\widetilde{g}} \SN{g,h}_{J}- Q(gh) \right)\\
\nonumber &+& (-1)^{\widetilde{f} \widetilde{g} + \widetilde{g} +1} Q(g) \left((-1)^{\widetilde{f}} \SN{f,h}_{J}- Q(fh)  \right).
\end{eqnarray}

\begin{remark}
The Loday--Poisson bracket is then a bi-derivation if we restrict the left-hand entries of the bracket to be $Q$-closed. However, this condition implies that the Loday--Poisson bracket is trivial. The other extreme is to insist that $(-1)^{\widetilde{f}} \SN{f,h}_{J}- Q(fh)=0$ for all $f,h \in C^{\infty}(M)$. This implies that $\SN{f,h}_{J}= (-1)^{\widetilde{f}}Q(fh)$ and thus the underlying odd Jacobi structure is $(0,Q)$. That is we have ``just" a Q-manifold.
\end{remark}

The similarity between the relations satisfied by the two classes of Hamiltonian vector field on an odd Jacobi manifold and the Cartan identities is striking, but not surprising as the Cartan calculus can be understood in terms of derived brackets \cite{Kosmann--Schwarzbach:2004}.  In essence we have the associations\\

\begin{tabular}{lcl}
Hamiltonian vector fields w.r.t odd Jacobi structure & $\longleftrightarrow$& Interior derivative\\
Hamiltonian vector fields w.r.t Loday--Poisson structure & $\longleftrightarrow$& Lie derivative\\
Loday--Poisson bracket & $\longleftrightarrow$& Lie bracket
\end{tabular}\\

The similarities can be summarised  in the table below:  \\

\renewcommand{\arraystretch}{2.0}
\begin{tabular}{|l||l|}
\hline
\textbf{Hamiltonian Vector Fields}  & \textbf{Cartan Identities}\\
\hline
$C^{\infty}(M) \rightarrow \Vect(M)$ & $\Vect(M) \rightarrow \Vect(\Pi TM)$ \\
$f \rightsquigarrow X_{f}$ &  $X \rightsquigarrow i_{X} $\\
\hline
$Y_{f} = - \left[Q, X_{f} \right]$ & $L_{X} = \left[d,i_{X} \right]$\\
\hline
$[Q, Y_{f}] =0$ & $[d,L_{X}] =0$\\
\hline
$\left[X_{f}, X_{g} \right] = - X_{\SN{f,g}_{J}}$ & $[i_{X}, i_{Y} ]=0$\\
\hline
$\left[ Y_{f} , X_{g} \right] = (-1)^{\widetilde{f}}X_{\{ f,g \}_{J}}$ & $\left[L_{X}, i_{Y} \right] = (-1)^{\widetilde{X}}i_{[X,Y]}$\\
\hline
$\left[Y_{f}, Y_{g}  \right] = Y_{\{f,g \}_{J}}$& $\left[L_{X}, L_{Y} \right]= L_{[X,Y]}$\\
\hline
\end{tabular}\\

With these formal algebraic similarities in mind, one can interpret the constructions here as a (partially noncommutative) generalisation of the Cartan calculus. However, as the interior product cannot directly be understood as a Hamiltonian vector field with respect to some odd Jacobi structure, the Cartan calculus cannot be seen as a special case of the constructions given in this work.

\section{The derived product}
On any Q-manifold  $(M,Q)$ one can consider the \emph{derived product}

\begin{equation}
f\ast g = (-1)^{\widetilde{f}+1} Q(f)g,
\end{equation}

where $f,g \in C^{\infty}(M)$. It is easy to verify that this derived product is associative, but not (super)commutative. The derived product is an \emph{odd} form of noncommutative multiplication on $C^{\infty}(M)$: $\widetilde{f\ast g} = \widetilde{f} + \widetilde{g} +1$. The notion of a derived product is also due to Loday  and have their origin in his study of dialgebras \cite{Loday:2001}.\\

The derived product on a Q-manifold can be viewed in the light of deformation quantisation.  That is, the vector space structure of the smooth functions on the Q-manifold remains the same, it is only the product that is \emph{deformed}. Also note that the derived $\ast$-commutator is given by

\begin{eqnarray}
\nonumber \left[f,g \right]_{\ast} &=& f \ast g - (-1)^{(\widetilde{f}+1)(\widetilde{g}+1)} g\ast f\\
\nonumber &=& - \SN{f,g}_{Q},
\end{eqnarray}

and thus, up to a sign, is  the odd Jacobi bracket generated by the homological vector field.  This is in the same spirit as understanding Poisson brackets as the classical limit of commutators in deformation quantisation. However, note that  $\Id \ast f =0$, meaning the constant function $\Id$ is not the identity (``bar-unit" in Loday's language) for the derived product. Also note $f \ast \Id = \pm Q(f)$. Furthermore, we  do not have any parameter playing the role of $\hbar$. \\
\newpage
\begin{remark}
We will not have any course in this work to employ ideas from the theory of dialgabras. We only remark that ``dialgebras are to Loday algebras what associative algebras are to Lie algebras". The relation between the $\ast$-commutator and  the odd Jacobi bracket on a Q-manifold is an example of this.
\end{remark}

\begin{lemma}\label{lemma:derived product}
Let $(M, S , Q)$ be an odd Jacobi manifold. Then the odd Jacobi bracket satisfies a generalised Leibniz rule given by
\begin{equation*}
\SN{f, g \ast h}_{J} = \SN{f,g}_{J}\ast h + (-1)^{(\widetilde{f}+1)(\widetilde{g}+1)} g \ast \SN{f,h}_{J} + (-1)^{\widetilde{f} + \widetilde{g} } f \ast g \ast h + (-1)^{\widetilde{g}} \{f,g \}_{J}h,
\end{equation*}
where $f,g,h \in C^{\infty}(M)$.
\end{lemma}

\begin{proof}
Direct from the definitions and the modified Leibniz rule for the odd Jacobi bracket we have
\begin{eqnarray}
\nonumber \SN{f, g \ast h}_{J} &=& (-1)^{\widetilde{g}+1}\SN{f, Q(g)h}_{J}\\
\nonumber &=& (-1)^{\widetilde{g}+1}\SN{f,Q(g)}_{J} h +(-1)^{\widetilde{g}+1 +(\widetilde{g}+1)(\widetilde{f}+1)}Q(g)\SN{f,h}_{J}- (-1)^{\widetilde{g}+1}\SN{f, \Id}_{J}Q(g)h.
\end{eqnarray}
Then using the fact that the homological vector field $Q$ is a Jacobi vector field the above can be cast in the form
\begin{eqnarray}
\nonumber \nonumber \SN{f, g \ast h}_{J} &=& (-1)^{\widetilde{f}+\widetilde{g}}Q\left(\SN{f,g}_{J} \right)h - -(1)^{\widetilde{f}+\widetilde{g}}\SN{Q(f), g}_{J}h\\
\nonumber &+& (-1)^{\widetilde{f}(\widetilde{g}+1)}Q(g)\SN{f,h}_{J} +( -1)^{\widetilde{g}+1}\SN{\Id,f}_{J}Q(g)h.
\end{eqnarray}
Then using the definitions  the lemma is established.
\end{proof}

\begin{theorem}\label{theo:derived product}
Let $(M, S , Q)$ be an odd Jacobi manifold. Then the Loday--Poisson bracket obeys the Leibniz with respect to the derived product:
\begin{equation*}
\{f, g \ast h \}_{J} = \{f,g \}_{J} \ast h + (-1)^{\widetilde{f} (\widetilde{g}+1)}g\ast \{f,h \}_{J},
\end{equation*}
where $f,g,h \in C^{\infty}(M)$.
\end{theorem}

\begin{proof}
The above theorem follows directly from Lemma \ref{lemma:derived product} upon the replacement $f \rightarrow Q(f)$ and the definition of the Loday--Poisson bracket.
\end{proof}

\noindent \textbf{Statement:} \begin{enumerate}\item\emph{Functions on an odd Jacobi manifold can be considered as elements of a noncommutative Loday--Poisson algebra with an odd form of multplication.}
\item \emph{Hamiltonian vector fields with respect to the Loday--Poisson bracket are derivations (from the left) over the derived product}.
\end{enumerate}
\newpage
\begin{proposition}
With the definitions previously given:
\begin{enumerate}
\item $X_{f \ast g} = (-1)^{\widetilde{f}+1}(f \ast \Id)X_{g} + (-1)^{\widetilde{f} \widetilde{g}}g X_{(f \ast\Id)} + (-1)^{\widetilde{f} +\widetilde{g}}(f \ast g)Q$,
\item $Y_{f \ast g} = (f \ast \Id)Y_{g} + (-1)^{(\widetilde{f}+1)(\widetilde{g}+1)}(g \ast \Id)Y_{f} - (f \ast g \ast \Id)Q$,
\end{enumerate}
where $f,g,h \in C^{\infty}(M)$.
\end{proposition}

\begin{proof}
The proof follows the definitions directly.
\begin{enumerate}
\item From Proposition \ref{prop:product}. we have
\begin{equation*}
X_{Q(f) g} = (-1)^{\widetilde{f}+1}Q(f)X_{g} + (-1)^{\widetilde{f}\widetilde{g}}g X_{Q(f)}+ (-1)^{\widetilde{f}+ \widetilde{g}} Q(f)g Q.
\end{equation*}
Then using the definition of the  derived product and $f \ast \Id = (-1)^{\widetilde{f}+1}Q(f)$ the first part of the proposition is established.
\item From Proposition \ref{prop:LodayProduct}. we have
\begin{equation*}
Y_{Q(f)g} = Q(f)Y_{g} + (-1)^{(\widetilde{f}+1)\widetilde{g} + \widetilde{g}+1}Q(g)\left(X_{Q(f)} +(-1)^{\widetilde{f} }Q(f)Q \right),
\end{equation*}
taking into account that $Q^{2}=0$. Then multiplying by the correct sign factor and using the definition of $Y_{f}$ produces
\begin{equation*}
Y_{f\ast g} = (f \ast \Id)Y_{g} + (-1)^{(\widetilde{f}+1)(\widetilde{g}+1)}(g \ast \Id)Y_{f}+ (-1)^{\widetilde{f} \widetilde{g}}Q(g)Q(f)Q.
\end{equation*}
Then using $Q(g)Q(f) = (-1)^{(\widetilde{f}+1)(\widetilde{g}+1)}  Q(f) Q(g) $ and the definition of the derived product the second part of the proposition is established.
\end{enumerate}
\end{proof}

\begin{remark}
As far as the author is aware, the study of Poisson-like brackets on algebras with an odd form of multiplication have not been studied in detail.
\end{remark}

\section{Jacobi algebroids}

An interesting class of odd Jacobi manifolds are the Jacobi algebroids  \cite{Grabowski:2001,Iglesias2001}. We apply some of the pervious constructions to the setting of Jacobi algebroids.

\subsection*{Definitions and basic results}

\begin{definition}
A vector bundle $E \rightarrow M$ is said to have the structure of a \textbf{Jacobi algebroid} if and only if the total space of $\Pi E^{*}$ comes equipped with a weight minus one odd Jacobi structure.
\end{definition}

It is well-known that Jacobi algebroids, which are also known as generalised Lie algebroids,  are equivalent to Lie algebroids in the presence of a 1-cocycle, see \cite{Iglesias2001}. Let us employ natural local coordinates $(x^{A}, \eta_{\alpha}, p_{A}, \pi^{\alpha})$ on the total space of $T^{*}(\Pi E^{*})$. The weight is assigned  as $\w(x^{A}) = 0$, $\w(p_{A})=0$, $\w(\eta_{\alpha}) = +1$ and $\w(\pi^{\alpha}) = -1$.  This is the \emph{natural weight} associated with the vector bundle structure $E^{*} \rightarrow M$.\\
 
The parity of the coordinates is given by $ \widetilde{x}^{A}=  \widetilde{A}$, $\widetilde{\eta}_{\alpha}= (\widetilde{\alpha} +1)$, $\widetilde{p}_{A}= \widetilde{A}$ and  $\widetilde{\pi}^{\alpha} =  (\widetilde{\alpha}+1)$. In these natural local coordinates the odd Jacobi structure is given by

\begin{eqnarray}
S &=&(-1)^{\widetilde{\alpha}}\pi^{\alpha}Q_{\alpha}^{A}(x)p_{A}+ (-1)^{\widetilde{\alpha} + \widetilde{\beta}}\frac{1}{2}\pi^{\alpha}\pi^{\beta}Q_{\beta \alpha}^{\gamma}\eta_{\gamma},\\
\nonumber  \Q &=& \pi^{\alpha}Q_{\alpha}(x),
\end{eqnarray}

which are both functions on the total space of $T^{*}(\Pi E^{*})$. The algebra of  ``multivector fields" $C^{\infty}(\Pi E^{*})$ comes equipped with an odd Jacobi bracket viz

\begin{equation}
\nonumber \SN{X,Y}_{E}  = (-1)^{\widetilde{X}+1}\{  \{S, X   \}_{T^{*}(\Pi E^{*})} , Y\}_{T^{*}(\Pi E^{*})}- (-1)^{\widetilde{X}+1} \{ \Q, XY \}_{T^{*}(\Pi E^{*})},
\end{equation}

with $X,Y \in C^{\infty}(\Pi E^{*})$. \\

In natural local coordinates this bracket is given by

\begin{eqnarray}\nonumber
\SN{X,Y}_{E} &=& Q_{\alpha}^{A}\left((-1)^{(\widetilde{X}+ \widetilde{\alpha}+1)(\widetilde{A}+1) } \frac{\partial X}{\partial \eta_{\alpha}}\frac{\partial Y}{\partial x^{A}}  - (-1)^{(\widetilde{X}+1)\widetilde{\alpha}} \frac{\partial X}{\partial x^{A}}\frac{\partial Y}{\partial \eta_{\alpha}}\right)\\
\nonumber &-& (-1)^{(\widetilde{X}+1)\widetilde{\alpha}+ \widetilde{\beta}}Q_{\alpha \beta}^{\gamma}\eta_{\gamma}\frac{\partial X}{\partial\eta_{\beta}}\frac{\partial Y}{\partial \eta_{\alpha}}\\
\nonumber &+&(-1)^{\widetilde{X}}Q_{\alpha}\frac{\partial X}{\partial \eta_{\alpha}} Y  + X Q_{\alpha}\frac{\partial Y}{\partial \eta_{\alpha}}.
\end{eqnarray}

Where $X = X(x, \eta) =  X(x) + X^{\alpha}(x) \eta_{\alpha} + \frac{1}{2!}X^{\alpha \beta}(x) \eta_{\beta}\eta_{\alpha} + \cdots$ \emph{etc}. Clearly, this odd Jacobi bracket is of weight minus one. If $Q_{\alpha} =0$, then the Jacobi algebroid reduces to a genuine Lie algebroid and the above bracket is a weight minus one Schouten bracket. The weight minus one odd Jacobi bracket is a natural generalisation of the weight minus one Schouten bracket associated with a Lie algebroid. \\







\subsection*{The Loday--Poisson brackets}
Now let us proceed to the Loday--Poisson bracket derived from the weight minus one odd Jacobi bracket and the homological vector field $Q = Q_{\alpha}\frac{\partial}{\partial \eta_{\alpha}}$. In natural  local coordinates the Loday--Poisson bracket is given by

\begin{eqnarray}
\nonumber \{X,Y \}_{E} &=& Q_{\alpha}^{A}\left(  (-1)^{\widetilde{A}(\widetilde{X} + \widetilde{\alpha})} Q_{\delta} \frac{\partial^{2}X}{\partial \eta_{\delta}\partial \eta_{\alpha}} \frac{\partial Y}{\partial x^{A}}  + (-1)^{\widetilde{X}(\widetilde{\alpha}+1)} \frac{\partial Q_{\delta}}{\partial x^{A}} \frac{\partial X}{\partial \eta_{\delta}} \frac{\partial  Y}{\partial \eta_{\alpha}}\right.\\
 \nonumber &+&\left. (-1)^{\widetilde{X}(\widetilde{\alpha}+1) + \widetilde{A}} Q_{\delta}\frac{\partial^{2}X}{\partial \eta_{\delta} \partial x^{A}}\frac{\partial Y}{\partial \eta_{\alpha}} \right) - (-1)^{\widetilde{X}(\widetilde{\alpha}+1)}Q_{\alpha \beta}^{\gamma}\eta_{\gamma}Q_{\delta} \frac{\partial^{2} X}{\partial \eta_{\delta} \partial \eta_{\beta}} \frac{\partial Y}{\partial \eta_{\alpha}}\\
 &+& (-1)^{\widetilde{\alpha}(\widetilde{X}+1)} Q_{\alpha}Q_{\beta}\frac{\partial X}{\partial \eta_{\beta}}\frac{\partial Y}{\partial \eta_{\alpha}}.
\end{eqnarray}

By construction, which is easily verified in natural local coordinates, the associated Loday--Poisson bracket is of weight minus two.

 \begin{remark}The Loday--Poisson bracket on $C^{\infty}(\Pi E^{*})$ should not be confused with the Poisson bracket on $C^{\infty}(E^{*})$ associated with the Lie algebroid structure ``behind" the Jacobi algebroid. Indeed, if we have a Lie algebroid and the trivial 1-cocycle $\mathcal{Q} =0$, then the associated Loday--Poisson bracket is obviously  itself trivial.
\end{remark}

\subsection*{The derived product}

The algebra of ``multivector fields" $C^{\infty}(\Pi E^{*})$ comes equipped  with a derived product viz

\begin{equation*}
X \ast Y = (-1)^{\widetilde{X}+1} Q(X)Y = (-1)^{\widetilde{X}+1}Q_{\alpha}\frac{\partial X}{\partial \eta_{\alpha}} \:\: Y.
\end{equation*}

Note that $f \ast X =0$, where $f \in C^{\infty}(M)$. Via Theorem \ref{theo:derived product}. the Loday--Poisson bracket satisfies the Leibniz rule from the left over the derived product:

\begin{equation*}
\{X, Y \ast Z \}_{E} = \{X,Y \}_{E} \ast Z + (-1)^{\widetilde{X} (\widetilde{Y}+1)}Y\ast \{X,Z \}_{E},
\end{equation*}

where $X,Y,Z \in C^{\infty}(\Pi E^{*})$.\\

\noindent \textbf{Statement:}\emph{``Multivector fields" on a Jacobi algebroid can be considered as elements of a noncommutative Loday--Poisson algebra with an odd form of multiplication}.










\section{Final remarks}

In this paper we used the derived bracket formalism to construct a Loday bracket on $C^{\infty}(M)$ from the initial datum of an odd Jacobi structure $(S,Q)$ on the supermanifold $M$. The Loday bracket in question is the bracket derived from the odd Jacobi bracket and the homological vector field $Q$. Furthermore it was shown that this Loday bracket satisfies the Leibniz rule from the left over the standard supercommutative product of functions and the derived product generated by the homological vector field. Thus the nomenclature \emph{Loday--Poisson bracket}. Some of the relations between the various Hamiltonian vector fields were also explored, as were some specific examples of Loday--Poisson brackets such as those in the theory of Lie algebroids and Jacobi algebroids.\\

Is is important to remark that the construction of the Loday--Poisson bracket from the odd Jacobi bracket presented here makes use of only half the structure available, namely just the homological vector field. This is the obvious thing to do if one wants to pass from an odd Jacobi bracket to an even Loday bracket. An obvious question here is \emph{can one do better and use the full odd Jacobi structure to pass from the odd Jacobi bracket to an even Loday bracket?}

\section*{Acknowledgements}
The author would like to thank Prof. J. Grabowski  and Dr R.A. Mehta for  their comments on  earlier drafts of this work.

\vfill
\begin{center}
Andrew James Bruce\\ \small \emph{Pembrokeshire College},\\
\small \emph{Haverfordwest, Pembrokeshire},\\\small  \emph{SA61 1SZ, UK}\\\small email:\texttt{andrewjamesbruce@googlemail.com}
\end{center}

\end{document}